\documentclass[pra,reprint,singlecolumn,superscriptaddress]{revtex4-2}
\usepackage{algorithm}
\usepackage{algpseudocode}
\usepackage{tikz}
\usetikzlibrary{shapes.geometric}
\usetikzlibrary{quantikz}
\usetikzlibrary{positioning}
\algrenewcommand\algorithmicindent{0.2cm}%

\usepackage{soul}
\usepackage{graphicx}% Figure package etc
\usepackage{float}% Float package, e.g., allows precise placement of figures
\usepackage{tabularx,booktabs} % Package for tables
\usepackage{tcolorbox}
\usepackage{amsmath}% Math package
\usepackage{amsthm}% Theorem package
\usepackage{bm}
\theoremstyle{remark}% Theorem style used by APS (remark,plain,definition)
\newtheorem*{conjecture*}{Conjecture}
\newtheorem*{proposition*}{Proposition}
\usepackage{mathrsfs} % Calligraphic math letters, e.g., \mathscr{X}
\usepackage{physics}% Physics package including Dirac notation
\usepackage{mathtools}% Package for mathematical symbols, e.g., \coloneqq
\usepackage{mathbbol}% Package for mathbbol font
\DeclareSymbolFontAlphabet{\mathbbol}{bbold}% Declares command \mathbbol
\usepackage{amssymb}% Package for mathbb font
\DeclareSymbolFontAlphabet{\mathbb}{AMSb}% Declares command \mathbb  (introduced because mathbbol package supersedes amssymb)
\usepackage{dsfont}% Package for mathds font, e.g., \mathds{X}
\DeclareMathAlphabet{\mathcal}{OMS}{cmsy}{m}{n}

	\usepackage{lipsum}% Dummy text package
	\usepackage{soul} % Strike out text using \st{...}
	\usepackage{comment} % Comments, e.g., \begin{comment}...\end{comment}
	\usepackage{tikz}
	\usetikzlibrary{quantikz}

	\usepackage{orcidlink}

	\graphicspath{{Figures/}}

	\algrenewcommand\algorithmicrequire{\textbf{Input:}}
	\algrenewcommand\algorithmicensure{\textbf{Output:}}

\begin{document}

	% Title
\title{On the equivalence between squeezing and entanglement potential for two-mode Gaussian states}

\author{Bohan Li \orcidlink{0000-0002-3928-8596}} %
\affiliation{Centre for Quantum Computation and Communication Technology, Department of Quantum Science and Technology, Australian National University, Canberra, ACT 2601, Australia}
\affiliation{These authors contributed equally}

	\author{Aritra Das \orcidlink{0000-0001-7840-5292}} %
	\affiliation{Centre for Quantum Computation and Communication Technology,
		Department of Quantum Science and Technology, Australian National University, Canberra, ACT 2601, Australia}
\affiliation{These authors contributed equally}

	\author{Spyros Tserkis \orcidlink{0000-0002-3272-990X}} %
	\affiliation{Physical Sciences, College of Letters and Science, University of California, Los Angeles (UCLA), CA, USA}

	\author{Prineha Narang \orcidlink{0000-0003-3956-4594}} %
	\affiliation{Physical Sciences, College of Letters and Science, University of California, Los Angeles (UCLA), CA, USA}

	\author{Ping Koy Lam \orcidlink{0000-0002-4421-601X}} %
	\affiliation{Centre for Quantum Computation and Communication Technology,
		Department of Quantum Science and Technology, Australian National University, Canberra, ACT 2601, Australia}
	\affiliation{Institute of Materials Research and Engineering, Agency for Science, Technology and Research (A*STAR), Singapore 138634, Republic of Singapore}

	\author{Syed M. Assad \orcidlink{0000-0002-5416-7098}} %
	\email{cqtsma@gmail.com, corresponding author}
	\affiliation{Centre for Quantum Computation and Communication Technology,
		Department of Quantum Science and Technology, Australian National University, Canberra, ACT 2601, Australia}
	\affiliation{School of Physical and Mathematical Sciences, Nanyang Technological University, Singapore 639673, Republic of Singapore}

	% Affiliations

	% Abstract
	%\section*{abstract}
	\begin{abstract}
		The maximum amount of entanglement achievable under passive transformations by continuous-variable states is called the entanglement potential. Recent work has demonstrated that the entanglement potential is upper-bounded by a simple function of the squeezing of formation, and that certain classes of two-mode Gaussian states can indeed saturate this bound, though saturability in the general case remains an open problem. In this study, we introduce a larger class of states that we prove saturates the bound, and we conjecture that all two-mode Gaussian states can be passively transformed into this class, meaning that for all two-mode Gaussian states, entanglement potential is equivalent to squeezing of formation. We provide an explicit algorithm for the passive transformations and perform extensive numerical testing of our claim, which seeks to unite the resource theories of two characteristic quantum properties of continuous-variable systems.
	\end{abstract}

		%\flushbottom
		\maketitle
		%\thispagestyle{empty}

		% Main Text
		\section{Introduction}\label{intro}

		Entanglement is a non-classical property that can be considered as a
		resource for various quantum technology
		applications~\cite{horodecki2009quantum}. In continuous-variable (CV)
		systems~\cite{WPGC+2012}, e.g., systems consisting of bosonic modes,
		entanglement is connected with a more fundamental property called
		squeezing~\cite{walls1983squeezed}.  Squeezing constitutes a necessary
		condition for entanglement in CV Gaussian
		systems~\cite{PhysRevA.64.063811, PhysRevA.65.032323,
			PhysRevA.66.024303} and finds applications in numerous areas of
		quantum optics and CV quantum information, including metrology~\cite{SqueezeMetro1,SqueezeMetro2,SqueezeMetro3}, secure
		quantum communication~\cite{SqueezeQComm,SqueezeQComm2,SqueezeQComm3,SqueezeQComm4}, quantum teleportation~\cite{SqueezeTeleport1,SqueezeTeleport2,SqueezeTeleport3}, cluster states~\cite{ExtCite1,ExtCite2}, heralded
		gates~\cite{SqueezeHeralded} and quantum computation~\cite{ExtCite3,ExtCite4}.

		Moreover,
		any multi-mode squeezed state can be transformed
		into an entangled state under passive operations~\cite{wolf2003entangling, Lami2018}.
		Passive operations in CV systems are relatively easier to perform in the laboratory than active operations.
		There exist multi-mode quantum states that are not entangled,
		but have the potential to be entangled
		by simply mixing on a beam splitter~\cite{PhysRevLett.94.173602,tserkis2020maximum}.
		Motivated by this,
		we study the entanglement potential of Gaussian states.
		Conceptually, the entanglement potential measures the maximum amount of entanglement obtainable under passive operations~\cite{PhysRevLett.94.173602}.
		This potential depends on the way that entanglement is measured, e.g.,
		in Ref.~\cite{PhysRevLett.94.173602},
		logarithmic negativity~\cite{vidal2002computable}
		was selected for this purpose,
		whereas in Ref.~\cite{tserkis2020maximum},
		the entanglement of formation~\cite{bennett1996mixed} was chosen.

		Focusing on the entanglement of formation,
		some of us have previously derived analytic expressions
		for the entanglement potential of a few specific classes of two-mode Gaussian states:
		symmetric states and balanced correlated states~\cite{tserkis2020maximum}.
		These analytic expressions were shown to be directly connected
		to
		the squeezing of formation~\cite{idel2016operational}---a
		measure that quantifies the amount of squeezing in a quantum state.
		In Ref.~\cite{tserkis2020maximum}, an explicit derivation of the passive operations needed to achieve this potential was provided.
		Further,
		it was shown that for general two-mode Gaussian states,
		a monotonic function,~$h_0(\cdot)$, of the squeezing of formation upper-bounds the entanglement potential.

		In this work, we extend that analysis  in two ways:
		first, we analytically show that for a larger, six-parameter class of two-mode Gaussian states,
		the entanglement potential is equal to $h_0(\cdot)$ of the squeezing of formation.
		Henceforth, we shall refer to all states
		having entanglement equal to entanglement potential
		as \textit{potential-saturating} states.
		Second, we conjecture that any two-mode Gaussian state can be passively transformed into a potential-saturating state
		from the six-parameter class of states,
		and present numerical evidence supporting this conjecture.
		If our conjecture holds true, then
		the entanglement potential of all two-mode Gaussian states is exactly equal to~$h_0(\cdot)$ of the squeezing of formation.
		In other words, we find that linear passive optics can
		always maximise the entanglement of a state
		up to a threshold value decided by
		the amount of squeezing present in the state.
		Our result, thus, connects
		the resource theories of squeezing and entanglement for two-mode Gaussian states
		and is
		primely relevant to quantum information and communication protocols, where squeezed states play a major role.

		Our paper is arranged as follows: In Sec.~\ref{pre}, we discuss
		some preliminaries of Gaussian quantum information.
		%including the definitions of squeezing of formation, entanglement of formation and entangelement potential.
		Then, in Sec.~\ref{bound} we introduce a special class of
		potential-saturating Gaussian states,
		and propose an algorithm
		to passively transform arbitrary two-mode Gaussian states
		into potential-saturating states.
		We present numerical simulations of our algorithm in
		Sec.~\ref{sec:results} to support our conjecture.
		Finally we conclude in Sec.~\ref{conclusion} with a discussion of our results and
		remarks on future scope.

		\section{Results}

		\subsection{Background}\label{pre}

		\subsubsection{Gaussian quantum information}

		%\paragraph{Gaussian QI:}
		%The Gaussian quantum toolbox lies
		%at the heart of
		%theoretical and experimental
		%quantum information for continuous variable (CV) systems~\cite{braunstein2005quantum}.
		Gaussian quantum states, which are the focus of this work, %~\cite{WPGC+2012}
		can be fully described %removed review reference: braunstein2005quantum,
		by the second statistical moments
		of the associated bosonic-field quadrature operators
		(assuming the first statistical moments,
		i.e., the mean values, to be zero).
		The quadrature field operators~$\hat{x}_j$ and~$\hat{p}_j$
		are the real and imaginary parts, respectively,
		of the bosonic-field annihilation operator for the~$j^\text{th}$ mode.
		Accordingly, any~$N$-mode Gaussian state admits
		a finite-dimensional representation via the covariance matrix~$\sigma$
		of its quadrature field operators.
		This covariance matrix is a~$2N\times 2N$ real symmetric matrix satisfying the uncertainty relation~\cite{simon1994quantum}~$\sigma+i \Omega \geq 0$,
		where~$\Omega$ is the symplectic form
		given in Appendix~\ref{sec:appA}.
		Apart from the regular eigenvalues~$\{\lambda_j\}$ of~$\sigma$,
		it is also useful to also define the symplectic eigenvalues~$\{\nu_j\}$ of~$\sigma$,
		which are the positive eigenvalues of~$i \Omega \sigma$.
		We denote the symplectic eigenvalues arranged in increasing order by~$\nu_j^\uparrow$,
		so that~$\nu_1^\uparrow \leq \nu_2^\uparrow \leq \nu_N^\uparrow$.
		Then,
		the uncertainty relation for~$\sigma$ is equivalent to
		the condition
		$\nu_1 \geq 1$~\cite{WPGC+2012}.

		In the symplectic representation, Gaussian transformations, which map Gaussian states to themselves,
		are given by symplectic matrices~$K \in \mathrm{Sp} (2 N, \mathbb{R})$,
		such that~$K \Omega K^\top = \Omega$, and $K$ acts on~$\sigma$ as~$\sigma \mapsto K \sigma K^\top$.
		Here~$\mathrm{Sp} (2 N, \mathbb{R})$ denotes the group of symplectic~$2N\!\times\!2N$ matrices over real numbers.
		Typical Gaussian transformations include beam splitters~$K_\text{bs}(\tau)$
		with transmissivity~$\tau\in[0,1]$
		%($\tau$ dropped for denoting a balanced beam splitter)
		and phase rotations~$K_\text{rot}(\theta)$
		with angle~$\theta\in[0, 2\pi)$;
		these are both passive operations,
		meaning they do not introduce extra energy into the system
		and thus, leave the trace of the covariance matrix,~$\Tr \sigma$ invariant.

		Active Gaussian transformations,
		on the other hand,
		include
		two local single-mode squeezers, denoted~$S_1(r_1, r_2)$,
		with real-valued squeezing parameters~$r_j$ for mode~$j\in\{1,2\}$
		or two-mode squeezers~$S_2(r)$ for~$r$ the single real squeezing parameter;
		these transformations introduce extra energy into the system.
		We summarise these transformations and their matrix representations
		in Appendix~\ref{sec:appA}.
		We also list a few standard decompositions
		in Gaussian quantum optics
		in Appendix~\ref{sec:appB};
		these will be used later in Secs.~\ref{bound} and~\ref{sec:results}.
		%The matrix representations of these transforms are standard and we mention them in the appendix for brevity.

		%%\paragraph{Pure State Decomposition and Convex Roof Construct:}
		The covariance matrix~$\pi$ of a pure Gaussian state
		satisfies~$\det\pi=1$, whereas for a mixed Gaussian state~$\sigma$, we
		have~$\det\sigma > 1$.  Such a mixed state $\sigma$ can be decomposed into
		a pure state $\pi$ and some positive definite matrix,~$\phi >0$,
		representing noise as~$\sigma = \pi + \phi$, but this decomposition is
		not unique. Owing to this non-uniqueness, one way to extend a
		resource measure $\mathcal{F}$ defined for pure states to mixed
		states is by optimising over all possible pure state decompositions as follows
		\begin{equation}
			\mathcal{F}(\sigma) \coloneqq \min_{\pi\color{red}}\left\{
			\mathcal{F}(\pi) \, \vert\,  \sigma-\pi \geq 0\, , \det \pi =1
			\right\} ,
			\label{eq:convexroof}
		\end{equation}
		where the minimisation is over all pure states~$\pi$.
		%Physically, this
		%corresponds to the minimum amount of resource $\mathcal{F}(\pi)$ required
		%to create the state $\sigma$.
		Below we discuss two resource
		measures defined in this way---the squeezing of
		formation~$\mathcal{S}(\sigma)$ and the entanglement of formation
		potential~$\mathcal{P}(\sigma)$ of a Gaussian state~$\sigma$.

		\subsubsection{Squeezing of formation}

		The process of squeezing a Gaussian state's uncertainty
		below the standard quantum limit~\cite{wu1986generation},
		along one quadrature,
		is an active transformation.
		Operational measures of squeezing have been proposed~\cite{idel2016operational} in order to quantify the amount of squeezing in a state.
		One such measure called the squeezing of formation (SOF), denoted~$\mathcal{S}(\sigma)$,
		is defined as the minimum amount of local squeezing
		required to construct~$\sigma$
		starting from vacuum~\cite{idel2016operational}.
		For an~$N$-mode pure Gaussian state~$\pi$, this quantity is simply a function of the eigenvalues of~$\pi$,
		\begin{equation}
			\mathcal{S}(\pi) \coloneqq -\frac{1}{2} \sum_{j=1}^N   \ln \left ( \lambda_j^{\uparrow} (\pi) \right )  = \frac{1}{2} \sum_{j=N+1}^{2N} \ln \left ( \lambda_j^{\uparrow} (\pi) \right ) ,
		\end{equation}
		where~$\lambda_j^\uparrow(\pi)$ denotes the~$j^\text{th}$ lowest eigenvalue of~$\pi$.
		Straightforwardly, the SOF of a
		two-mode locally-squeezed vacuum
		with squeezing parameters~$r_1$ and~$r_2$
		is simply~$r_1+r_2$.
		Finally, for mixed states~$\sigma$, the SOF definition is then extended via
		\begin{equation}
			\mathcal{S}(\sigma) \coloneqq  \min_{\pi}\{ \mathcal{S}(\pi)
			\, \vert\,  \sigma-\pi \geq 0\,, \det \pi =1\},
			\label{eq:defSOF}
		\end{equation}
		where the minimisation is over all pure states~$\pi$.
		\subsubsection{Entanglement of formation potential}
		\label{subsec:EOFP}

		A two-mode Gaussian~$\sigma$ is separable if and only if its partial transpose, denoted~$\sigma^\Gamma$, is also a valid state, i.e.,
		\begin{equation}
			\nu_1^\uparrow(\sigma^\Gamma) \geq 1 ,
		\end{equation}
		a result known as the PPT condition~\cite{GKLC01}.
		In this case,
		$\sigma$ has zero entanglement
		irrespective of which entanglement measure is employed.
		However,
		for mixed entangled states,
		the various measures of entanglement,
		including logarithmic negativity~\cite{vidal2002computable}, entanglement of formation~\cite{bennett1996mixed}, distillable entanglement~\cite{bennett1996mixed}, and relative entropy of entanglement~\cite{VPRK97},
		are all in general inequivalent~\cite{AI05, VP00}.
		%By convention, we choose~$\sigma^\gamma$ to be the partial transposition of the second mode of $\sigma$, specifically, $\sigma^\gamma = P \sigma P$ with~$P=\textsc{Diag}(1,1,1,-1)$.

		%%\paragraph{EOF:}
		We limit our scope to two-mode Gaussian states,
		which can be treated as a bipartite system,
		and we choose the entanglement of formation (EOF), denoted~$\mathcal{E}(\sigma)$,
		as our entanglement measure~\cite{bennett1996mixed}.
		Conceptually,~$\mathcal{E}(\sigma)$ quantifies the minimum amount of entanglement required
		to produce the state~$\sigma$, assisted only by local operations and classical communication (LOCC).
		For pure states~$\pi$,~$\mathcal{E}(\pi)$ is defined to be the entropy of entanglement~\cite{HSH99, AI05},
		i.e.,
		\begin{equation}
			\mathcal{E}(\pi)\coloneqq \max\left\{0, h\left[\nu_1^{\uparrow} (\pi^\Gamma)\right]\right\},
		\end{equation}
		where~$h[\cdot]$ is an auxiliary function defined in Appendix~\ref{sec:appC}.
		Then, for mixed states~$\sigma$,
		the definition
		is extended via Eq.~\eqref{eq:convexroof} to~\cite{MM08, TR17, TOR19, WGKW+04}
		\begin{equation}
			\mathcal{E}(\sigma) \coloneqq  \min_{\pi}\{ \mathcal{E}(\pi)
			\, \vert\,  \sigma-\pi \geq 0\,, \det \pi=1\}.
			\label{eq:EOFDef}
		\end{equation}
		%\doubt{Notably, for pure de-cross-correlated states
			%(states with~$\langle \hat{x}_i \hat{p}_j\rangle=\langle \hat{p}_i \hat{x}_j \rangle=0$ for~$i, j \in \{1, 2\}$~\cite{AGLL20,tserkis2020maximum, DGCZ00}),
			%the EOF becomes a monotonically increasing function
			%of the local covariance matrix~\cite{AI05,WGKW+04};
			%see Appendix~\ref{sec:appA} for details.}
		Note that Eq.~\eqref{eq:EOFDef} technically defines the Gaussian-EOF~\cite{WGKW+04}, which, in general, upper-bounds the EOF for multi-mode states, but coincides with the EOF for two-mode Gaussian states~\cite{akbari2015entanglement}.

		Next, the EOF potential $\mathcal{P}$ is defined as the maximum EOF a
		state can attain when transformed only by passive linear
		optics~\cite{tserkis2020maximum}. Specifically, starting from a two-mode Gaussian
		state~$\sigma$, with access to two ancillary vacuum modes, and
		four-mode passive transformations~$K$, the EOF potential is defined as
		\begin{equation}
			\mathcal{P}(\sigma) \coloneqq  \sup_{K} \left \{ \mathcal{E} ( \sigma') \, {\bigg \vert} \, \sigma' = \tr_2 \left [K (\sigma \oplus \mathds{1}_2) K^\top \right]\right \},
			\label{eq:EOFPdefn}
		\end{equation}
		so that~$\mathcal{E}(\sigma) \leq \mathcal{P}(\sigma)$ always.  In
		Eq.~\eqref{eq:EOFPdefn}, the~$\mathds{1}_2$ denotes two ancillary
		vacuum modes and the~$\tr_2$ denotes tracing out these
		modes.
		%Intuitively, the EOF potential denotes the maximum EOF we can
		%get from the state $\sigma$ by rearranging it between the four modes,
		%two original modes plus two ancillary modes.
		Interestingly,~$\mathcal{P}(\sigma) $ is upper-bounded by a simple
		function of~$\mathcal{S}(\sigma)$~\cite{tserkis2020maximum},
		%In Ref.~\cite{tserkis2020maximum}, it is shown that the maximum attainable EOF using passive operations and ancillary vacuum modes, defined as the EOF potential, is upper bounded by a simple function of SOF as
		\begin{equation}
			\mathcal{P}(\sigma) \leq h_0\left [\mathcal{S}(\sigma) \right ],
			\label{eq:defBOUND}
		\end{equation}
		where $h_0[\cdot]$ is a monotonic auxiliary function defined in Appendix~\ref{sec:appC}.
		% for simplification defined as $h_0(x)=h(e^{-x})$.
		%This inequality relates the EOF and the SOF by a simple formula. In addition, it provides indication on the maximum attainable entanglement without adding extra energy into the system.
		However, the saturability of the bound in Eq.~\eqref{eq:defBOUND} for arbitrary~$\sigma$ remains an open problem.
		In this work, we
		provide an algorithm
		that aims to saturate this bound for arbitrary two-mode Gaussian states
		and then establish this saturability via extensive numerical testing.

		%\clearpage

		\subsection{Saturating the EOF Potential}\label{bound}
		%\section{Saturating the EOF Upper Bound}

		In this section, we first introduce a special class of
		potential-saturating two-mode Gaussian states,~$\sigma_\mathrm{sp}$
		($\mathrm{sp}$ for special),
		which have~$\mathcal{E}(\sigma_\mathrm{sp})=\mathcal{P}(\sigma_\mathrm{sp})=h_0\left [\mathcal{S}(\sigma_\mathrm{sp})\right]$,
		and thus saturate the bound in Eq.~\eqref{eq:defBOUND}.
		We state this claim as a proposition and then prove it in Sec.~\ref{subsec:specialclassofstates}.
		Then,
		in Sec.~\ref{arbitrary}, we conjecture that any arbitrary two-mode Gaussian state can be passively transformed into this special class.
		%Our conjecture is presented without proof but
		In Sec.~\ref{subsec:unifiedalgorithm} we provide an explicit algorithm to perform this transformation.
		If our conjecture holds true,
		then~$\mathcal{P}(\sigma) = h_0\left [ \mathcal{S}(\sigma) \right]$ for all two-mode Gaussian states.

		\subsubsection{A special class of states}
		\label{subsec:specialclassofstates}
		Consider the two-mode Gaussian state
		\begin{equation}
			\sigma_\text{sp} =K_\text{bs} \big( \pi_d(r_1,r_2)+\lambda_1\phi_1+\lambda_2\phi_2\big) K_\text{bs}^\top \,,
			\label{eqn:sf}
		\end{equation}
		where~$\pi_\text{d}(r_1,r_2)$ represents a locally-squeezed two-mode
		pure state in diagonal form with squeezing parameters $r_1$ and $r_2$
		(matrix representation in Appendix~\ref{sec:appA}).
		Here $K_\text{bs}$
		denotes a balanced beam splitter operation with $\tau=1/2$,
		$\lambda_2 \geq \lambda_1 \geq 0$ are two  non-negative constants,
		and~$\phi_1=\ketbra{\phi_1}$
		and $\phi_2=\ketbra{\phi_2}$ are two orthogonal, positive
		semidefinite, rank-one matrices with
		\begin{align}
			\ket{\phi_1} =\begin{pmatrix*}[c]  \alpha \cos \theta\\\sin \theta\\
				\cos \theta\\
				\alpha  \sin \theta\end{pmatrix*}\;\text{and}\;
			\ket{\phi_2} =\begin{pmatrix*}[c]  \alpha \sin\theta \\-\cos\theta\\
				\sin\theta\\ - \alpha \cos\theta \end{pmatrix*}\,.
			\label{eq:noiseprojectormatrices}
		\end{align}
		In Eq.~\eqref{eq:noiseprojectormatrices},~$\alpha$ is a real parameter
		satisfying~$\vert \alpha \vert \leq e^{-r_1-r_2}$
		and~$\theta\in[0, 2\pi)$ is an angle.  The
		term~$\lambda_1 \phi_1 + \lambda_2 \phi_2$ can be thought of as
		correlated noise, parameterised by $\lambda_1$, $\lambda_2$, $\alpha$
		and $\theta$, added to the pure two-mode squeezed state $\pi_d$.  The
		terms $\lambda_1$ and~$\lambda_2$ denote the strength of the noise
		terms~$\phi_1$ and~$\phi_2$, respectively.
		The parameter~$\alpha$ determines the ratio
		between the added noise in the first and the second modes in the same quadrature, whereas the angle
		$\theta$ determines the ratio between the added noise in the $\hat{x}$ and $\hat{p}$
		quadratures in the same mode.  When $\lambda_1=\lambda_2$, the form of
		the added noise $\lambda_1 \phi_1 +\lambda \phi_2$ is special in the
		sense that the state~$\sigma_{\mathrm{sp}}$ becomes passively
		de-cross-correlatable, i.e., can be passively transformed into a
		de-cross-correlated state (recall that de-cross-correlated states have
		no correlations between the~$\hat{x}$ and~$\hat{p}$ quadratures, i.e.,
		$\langle \hat{x}_i \hat{p}_j + \hat{p}_i \hat{x}_j \rangle=0$
		for~$i, j \in \{1, 2\}$, see Appendix~\ref{sec:appA} for
		details). Overall, the state~$\sigma_\mathrm{sp}$ has~6 free
		parameters~$\{r_1, r_2, \lambda_1, \lambda_2, \alpha, \theta\}$ and
		thus may be thought of as an element from a six-parameter family of
		states.

		As we shall show in the following proposition, the state
		$\sigma_\mathrm{sp}$ is special in the sense that: $\sigma_\mathrm{sp}$ has the same SOF
		as $\pi_d$, the EOF of $\sigma_\mathrm{sp}$ saturates its EOF
		potential and $\sigma_\mathrm{sp}$ has the same EOF potential as $\pi_d$:
		\begin{equation}
			\mathcal{S}(\sigma_\text{sp}) =\mathcal{S}(\pi_d) \; \;
			\text{and} \; \; \mathcal{E}(\sigma_\text{sp}) = \mathcal{P}(\sigma_\text{sp}) = \mathcal{P}(\pi_d).
			\label{eq:specialSOFEOFequalpure}
		\end{equation}
		Moreover, the EOF and SOF properties of a pure state~$\pi_\text{d}$ are simply
		\begin{equation}
			\mathcal{S}(\pi_\text{d}) = r_1+r_2\quad \text{and}\quad \mathcal{P}(\pi_\text{d}) = h_0 \left [\mathcal{S}(\pi_d) \right ].
			\label{eq:EOFSOFpurediag}
		\end{equation}
		In other words, the upper bound for EOF in Eq.~\eqref{eq:defBOUND} is saturated for all such~$\sigma_\mathrm{sp}$.
		We now formally state and prove this claim.

		\begin{figure*}[btp]
			\centering
			\includegraphics[width=0.95\textwidth]{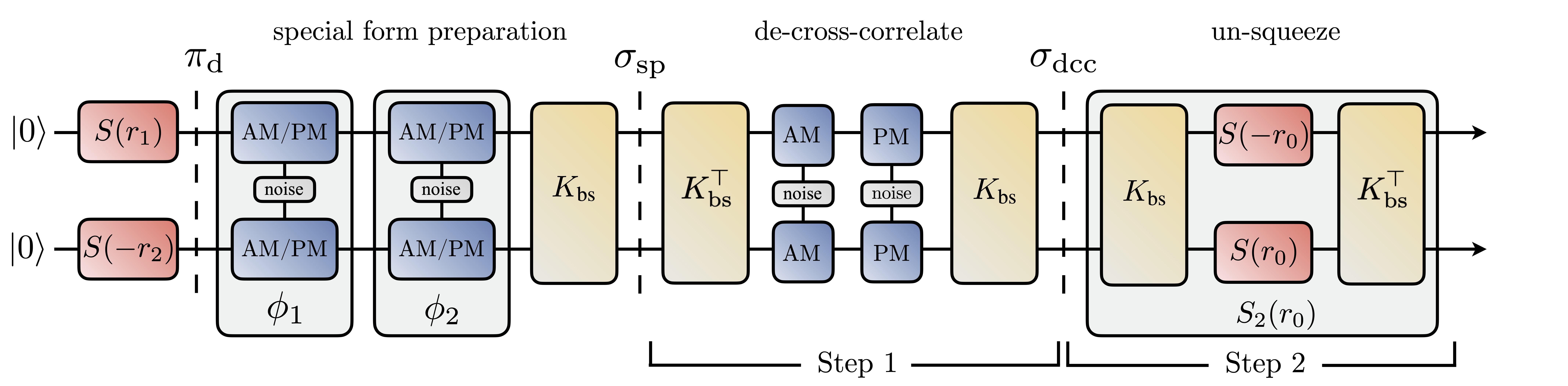}
			\caption{Schematic of the procedure to compute the EOF
				potential~$\mathcal{P}$ for a state $\sigma_\text{sp}$ in the special form given in Eq.~\eqref{eqn:sf}.
				Steps~1 and~2 from the proof of our proposition are also indicated.
				After adding a particular correlated noise to~$\sigma_\mathrm{sp}$ (step 1),
				the de-cross-correlated state~$\sigma_\mathrm{dcc}$ is then
				two-mode-squeezed to produce a separable state (step 2).
				The minimum value~$r_0$ of the two-mode squeezing parameter, such that the output state is separable, yields the lower bound~$h_0[2 r_0]$ to~$\mathcal{P}(\sigma_\mathrm{sp})$, as in Eq.~\eqref{eq:boundgreaterthan}.}
			\label{fig:2ms_schematic}
		\end{figure*}

		\begin{proposition*}
			\label{proposition1}
			For any state $\sigma_\mathrm{sp}$ of the form in Eq.~\eqref{eqn:sf},
			the EOF upper bound in Eq.~\eqref{eq:defBOUND} is saturated, i.e.,
			\begin{equation}
				\mathcal{E}(\sigma_\mathrm{sp})
				= \mathcal{P}(\sigma_\mathrm{sp})
				%=h_0(r_1+r_2)
				=h_0 \left [\mathcal{S}(\sigma_\mathrm{sp})\right]\; .
			\end{equation}
		\end{proposition*}

		\begin{proof}
			The outline of our proof is as follows.
			By adding classical correlations in the form of noise to~$\sigma_\mathrm{sp}$,
			we get a state~$\sigma_\mathrm{dcc}$ that is de-cross-correlated.
			We then lower-bound~$\mathcal{E}(\sigma_\mathrm{dcc})$,
			which serves as a lower bound for~$\mathcal{E}(\sigma_\mathrm{sp})$
			and thus~$\mathcal{P}(\sigma_\mathrm{sp})$.
			Finally, we upper-bound~$\mathcal{S}(\sigma_\mathrm{sp})$ and show that
			this upper bound coincides with the lower bound for~$\mathcal{P}$,
			which along with Eq.~\eqref{eq:defBOUND} implies that~$\mathcal{P}(\sigma_\mathrm{sp}) = h_0\left[\mathcal{S}(\sigma_\mathrm{sp})\right]$.
			The proof presented below is broken up into three steps,
			and is illustrated in Fig.~\ref{fig:2ms_schematic} as a circuit diagram.

			Step 1:
			We first add some noise along~$K_\mathrm{bs} \phi_1  K_\mathrm{bs}^\top$
			to~$\sigma_\mathrm{sp}$ to get a de-cross-correlated state~$\sigma_\mathrm{dcc}$,
			\begin{align*}
				\sigma_\mathrm{dcc}& =\sigma_\text{sp}+(\lambda_2-\lambda_1)K_\text{bs} \phi_1 K_\text{bs}^\top \\
				&= K_\text{bs}(\pi_\text{d} + \lambda_2(\phi_1+\phi_2))K_\text{bs}^\top \\
				&= K_\text{bs} \left ( \pi_d + \lambda_2
				\begin{pmatrix}\alpha^2&0&\alpha&0\\
					0&1&0&\alpha\\
					\alpha&0&1&0\\
					0&\alpha&0&\alpha^2\end{pmatrix} \right)K_\text{bs}^\top\;.
			\end{align*}
			As adding noise cannot increase entanglement, we have
			\begin{equation}
				\label{eq:adding_noise}
				\mathcal{E}(\sigma_\mathrm{dcc}) \leq \mathcal{E}(\sigma_\text{sp}).
			\end{equation}

			Step 2: Next, we consider the least amount of two-mode
			squeezing,~$r_0$, required to un-squeeze~$\sigma_\mathrm{dcc}$
			into a separable state, i.e.,
			\begin{align}
				r_0 \coloneqq  \min_{r}\left\{
				r \Bigm| S_2(r) \sigma_\mathrm{dcc} S_2^\top(r) \;\text{separable}\right\}  \;.
			\end{align}
			Then~$h_0\left[2 r_0\right]$ is a lower bound to~$\mathcal{E}(\sigma_\mathrm{dcc})$.
			By checking the necessary and sufficient conditions for separability
			(see Sec.~\ref{subsec:EOFP}),
			we find that the state~$S_2(r)\sigma_\mathrm{dcc} S_2^\top(r)$
			is separable when
			\begin{equation}
				\frac{r_1 + r_2}{2} \leq r \leq \frac{1}{4}\sum_{j=1}^2 \log\left [\frac{\lambda_2+e^{2r_j}}{1+\lambda_2 \alpha^2 e^{2r_j}} \right ],
				\label{eq:rmin}
			\end{equation}
			so that~$r_0 = (r_1+r_2)/2$.
			Moreover, for the interval in Eq.~\eqref{eq:rmin} to be valid,
			we must have~$\vert \alpha \vert \leq e^{-r_1-r_2}$.
			The lower bound~$h_0[2 r_0] = h_0[r_1+r_2] \leq \mathcal{E}(\sigma_\mathrm{dcc})$
			from Eq.~\eqref{eq:rmin},
			when combined with Eq.~\eqref{eq:adding_noise},
			results in
			\begin{equation}
				h_0[r_1+r_2] \leq \mathcal{E}(\sigma_\mathrm{dcc}) \leq \mathcal{E}(\sigma_\text{sp}) \leq \mathcal{P}(\sigma_\text{sp})  .
				\label{eq:boundgreaterthan}
			\end{equation}
			Step 3:
			Finally, we observe that~$\sigma_\mathrm{sp}$ can clearly
			be produced with~$r_1+r_2$ amount of squeezing,
			so that $\mathcal{S}(\sigma_\text{sp})\leq r_1+r_2$.
			The monotonicity of~$h_0(\cdot)$ and Eq.~\eqref{eq:defBOUND}
			then allows us to upper-bound~$\mathcal{E}(\sigma_\mathrm{sp})$
			as
			\begin{equation}
				\mathcal{E}(\sigma_\text{sp})\leq  \mathcal{P}(\sigma_\mathrm{sp}) \leq h_0 \left [\mathcal{S}(\sigma_\text{sp})\right]\leq h_0[r_1+r_2].
				\label{eq:boundlessthan}
			\end{equation}
			Combining Eqs.~\eqref{eq:boundgreaterthan} and~\eqref{eq:boundlessthan}, we get
			\begin{align}
				\mathcal{E}(\sigma_\text{sp})=\mathcal{P}(\sigma_\text{sp})=h_0[r_1+r_2]\;\text{and}\;\mathcal{S}(\sigma_\text{sp})=r_1+r_2
				\,,
			\end{align}
			thus proving the proposition.
			%\end{enumerate}
		\end{proof}

		The proposition above says that for states in the special form of Eq.~\eqref{eqn:sf},
		the upper bound~$h_0\left[\mathcal{S}(\cdot)\right]$
		(introduced in Ref.~\cite{tserkis2020maximum})
		on the EOF potential~$\mathcal{P}(\cdot)$
		is actually the true value of~$\mathcal{P}(\cdot)$.
		In other words, all states in this six-parameter family
		saturate the inequality in Eq.~\eqref{eq:defBOUND}.
		Notably, previously, only two three-parameter families
		of two-mode Gaussian states were known to possess this property:
		symmetric states and balanced correlated states~\cite{tserkis2020maximum}.

		\subsubsection{Extension to all two-mode Gaussians}\label{arbitrary}
		%\section{Extension to arbitrary two-mode Gaussians}\label{arbitrary}

		Let us now denote by~$G$ the set of all states in the special form of Eq.~\eqref{eqn:sf}.
		Suppose a state~$\sigma'$ is not in this set~$G$,
		but on applying some passive transformation~$K'$
		transforms into the special form,
		i.e.,
		\begin{equation}
			\sigma'\not\in G \quad \mathrm{but} \quad K' \sigma' K'^\top \in G.
		\end{equation}
		As passive transformations by definition do not change the EOF potential of a state~\cite{tserkis2020maximum},
		we must have
		\begin{equation}
			\mathcal{P}(\sigma')= \mathcal{P}(K'\sigma'K'^\top) = h_0 \left [ \mathcal{S}(K'\sigma'K'^\top) \right ].
		\end{equation}
		Moreover, passive transformations also leave the SOF invariant~\cite{idel2016operational},
		so~$\mathcal{S}(K'\sigma'K'^\top) = \mathcal{S}(\sigma')$. Thus, we have
		\begin{equation}
			\mathcal{P}(\sigma') = h_0 \left [ \mathcal{S}(\sigma')\right ],
			\label{eq:boundsatforsigmadashed}
		\end{equation}
		indicating that~$\sigma'$ too
		saturates the inequality in Eq.~\eqref{eq:defBOUND}
		despite
		not being in the set~$G$.  By a similar line of reasoning, it follows
		that for any state~$\sigma'\not\in G$, if we can add  some
		noise~$\phi'$ such that its SOF remains unchanged, i.e.,~$\mathcal{S}(\sigma')
		= \mathcal{S}(\sigma'+\phi')$, and the resulting state is in the
		special form, i.e.,~$\sigma'+\phi' \in G$, then $\sigma'$ must
		also satisfy Eq.~\eqref{eq:boundsatforsigmadashed}.

		It is then evident that any state that can be transformed into~$G$ by
		either passive transformations, or the addition of noise that keeps
		the SOF constant, or both,
		must also saturate the upper bound
		in Eq.~\eqref{eq:defBOUND}.
		We
		conjecture that all
		two-mode Gaussian states can be transformed into~$G$ in this way.

		\begin{conjecture*}
			Any two-mode Gaussian state~$\sigma_\mathrm{in}$ can be
			transformed into some element~$\sigma_\mathrm{out}$ in~$G$,
			without increasing its SOF, via only passive transformations, the addition of noise and
			access to ancillary vacuum modes.
			\label{conjecture1}
		\end{conjecture*}

		From the discussion in Sec.~\ref{subsec:specialclassofstates},
		we know that our conjecture, if true,
		would immediately imply that~$\mathcal{P}(\sigma_\mathrm{in})=h_0\left [ \mathcal{S}(\sigma_\mathrm{in})\right ]$
		for all two-mode Gaussian states~$\sigma_\mathrm{in}$.
		In this work, we do not formally prove our conjecture---instead,
		we provide evidence for the conjecture in the following way.
		First we present the transformation~$\sigma_\mathrm{in} \mapsto \sigma_\mathrm{out}$
		mentioned in the conjecture
		%This transformation~$K_\sigma$ of course depends on~$\sigma$,
		%and thus we specify it
		as an algorithm (Alg.~\ref{algorithm:MaximizeEOF} in Sec.~\ref{subsec:unifiedalgorithm}).
		Algorithm~\ref{algorithm:MaximizeEOF} takes~$\sigma_\mathrm{in}$ as input,
		and after performing passive operations, adding noise, and adding and then discarding an ancillary vacuum mode,
		the algorithm outputs the transformed state~$\sigma_\mathrm{out} \in G$.
		Next, we numerically ran our algorithm on~$10^6$ random inputs~$\sigma_\mathrm{in}$,
		and calculate the EOF of the output $\mathcal{E}(\sigma_\text{out})$ and compare
		that to the SOF of the input $\mathcal{S}(\sigma_\text{in})$.
		We verify that ~$\mathcal{E}(\sigma_\mathrm{out})=h_0[
		\mathcal{S}(\sigma_\mathrm{in})]$ holds true for every input
		state to within numerical tolerances.

		\subsubsection{Algorithm: Passive operations to maximise EOF}
		\label{subsec:unifiedalgorithm}
		We now propose an algorithm that, starting from any arbitrary two-mode
		Gaussian~$\sigma_\mathrm{in}$,
		outputs a potential-saturating
		two-mode Gaussian~$\sigma_\mathrm{out}$
		such that~$\mathcal{E}(\sigma_\mathrm{out})=\mathcal{P}(\sigma_\mathrm{out})=h_0
		\left [\mathcal{S}(\sigma_\mathrm{out}) \right ]$
		while keeping the SOF
		constant,
		i.e.,
		$\mathcal{S}(\sigma_\text{out})=\mathcal{S}(\sigma_\text{in})$.
		In doing so, the
		algorithm only performs passive operations and adds noise to the input
		state so
		that~$\mathcal{P}(\sigma_\mathrm{out}) \leq
		\mathcal{P}(\sigma_\mathrm{in})$.
		%and~$\mathcal{S}(\sigma_\mathrm{in}) = \mathcal{S}(\sigma_\mathrm{out})$.
		As a result, our algorithm establishes the fact that
		$\mathcal{P}(\sigma_\mathrm{in})= h_0 \left [ \mathcal{S}(\sigma_\mathrm{in})\right ]$
		for any arbitrary two-mode Gaussian~$\sigma_\mathrm{in}$.
		The fundamental idea behind the algorithm is to
		decouple the squeezing between the two modes of~$\sigma_\mathrm{in}$,
		and then mix the two modes on a balanced beam splitter.
		The resulting de-cross-correlated state, with two identical modes,
		is known to be potential-saturating and also saturates the EOF bound
		in Eq.~\eqref{eq:defBOUND}
		(see Appendix~\ref{sec:appC}).

		The first step in the algorithm is
		to find an optimal pure state~$\pi_\mathrm{opt}$
		that has the same SOF as~$\sigma_\mathrm{in}$ from Eq.~\eqref{eq:defSOF},
		i.e.,~$\mathcal{S}(\sigma_\mathrm{in}) = \mathcal{S}(\pi_\mathrm{opt})$;
		in Alg.~\ref{algorithm:MaximizeEOF}, we denote this procedure as~\textsc{OptSOFState}$(\sigma_\mathrm{in})$~\cite{idel2016operational}.
		Next,~\textsc{BMDecomp}$(\pi_\mathrm{opt})$ leverages the Bloch-Messiah decomposition to find
		a passive transformation~$K_\mathrm{BM}$ that
		that diagonalises~$\pi_\mathrm{opt}$ to~$\pi_\mathrm{diag}$
		(see Appendix~\ref{sec:appB} for details).
		Applying~$K_\mathrm{BM}$ to the mixed state~$\sigma_{\mathrm{in}}=\pi_\mathrm{opt}+\phi$
		yields the mixed state~$\sigma_{\mathrm{diag}}=\pi_\mathrm{diag} + \phi_\mathrm{diag}$
		(note that~$\sigma_{\mathrm{diag}}$ and~$\phi_{\mathrm{diag}}$ are not diagonal).
		In the second step, we calculate the eigenvalues~$\{\lambda_j\}$
		(arranged in decreasing order)
		and eigenvectors~$\{\ket{\phi_j}\}$ of the matrix~$\phi_\mathrm{diag}$
		via the procedure~\textsc{Spectrum}$(\phi_\mathrm{diag})$.
		Then we compute the extra noise term~$\phi_\mathrm{extra} = (\lambda_1 - \lambda_2) \ketbra{\phi_2}{\phi_2}$,
		which, when added to~$\sigma_\mathrm{diag}$,
		gives us the state~$\sigma'=\sigma_\mathrm{diag}+\phi_\mathrm{extra}$.

		Surprisingly, we find that the state~$\sigma'$
		at this point in the algorithm
		can always be passively de-cross-correlated.
		This is not true, in general, for mixed Gaussian states.
		Nevertheless, for all~$\sigma_\mathrm{in}$,
		$K_\mathrm{BM} \sigma_\mathrm{in} K_\mathrm{BM}^\top + \phi_\mathrm{extra}$
		becomes a passively de-cross-correlatable state---this
		is crucial because de-cross-correlated states are
		optimal for the EOF potential (see Appendix~\ref{sec:appC}).
		This passive transformation,
		which is simply a phase rotation on one mode,
		is calculated in the procedure~\textsc{DeCrossCorrelate}$(\cdot)$
		by numerically finding the angle~$\theta^*\in[0, 2 \pi)$ and mode~$i^*\in\{1, 2\}$ to be rotated
		to make~$\sigma'$ de-cross-correlated.
		The last step in the algorithm comprises mixing one of the modes
		of the de-cross-correlated state~$\sigma_\mathrm{rot}$
		with a third ancillary vacuum mode on a beam splitter;
		this is done to remove noise from~$\sigma_\mathrm{rot}$.
		The transmissivity~$\tau^*\in[0, 1]$ for this beam splitter
		operation~$K_\mathrm{bs}^{3, j^*}$
		and the mode~$j^*\in\{1,2\}$ to be mixed with vacuum are
		calculated numerically by maximising the EOF
		of the resulting state.
		Details of the numerical procedure for calculating EOF
		are presented in Sec.~\ref{sec:results}.
		This final state is output as~$\sigma_\mathrm{out}$
		by Alg.~\ref{algorithm:MaximizeEOF}, which we present below in full.

		\begin{algorithm}[H]
			\caption{Maximizing EOF}
			\label{algorithm:MaximizeEOF}
			\begin{algorithmic}[1]
				\Require
				\Statex $\sigma_{\mathrm{in}}$
				\Comment Two-mode mixed Gaussian state\;
				\Ensure
				\Statex $\sigma_{\mathrm{out}}$
				\Comment State which saturates SOF-EOF bound
				\Procedure{\textsc{MaxEOF}}{$\sigma_\mathrm{in}$}
				%		\State $\pi_\mathrm{opt} \gets \arg \min \{ \mathcal{S}(\pi) \vert \det(\pi)=1, \sigma_\mathrm{in} - \pi \geq 0 \}  $\;
				\State $\pi_\mathrm{opt} \gets \textsc{OptSOFState}(\sigma_\mathrm{in}) $\;
				\State $\{K_\mathrm{BM}, \pi_\mathrm{diag}\} \gets \textsc{BMDecomp}(\pi_\mathrm{opt})$
				%		\Comment $K_\mathrm{BM} \pi_\mathrm{opt} K_\mathrm{BM}^\top = \pi_\mathrm{diag}$\;
				\State $\sigma_\mathrm{diag} \gets K_\mathrm{BM} \sigma_\mathrm{in} K_\mathrm{BM}^\top$\;
				\State $\phi_\mathrm{diag} \gets \sigma_\mathrm{diag} - \pi_\mathrm{diag}$\; \label{step:5}
				\State $\{\lambda_j, \ket{\phi_j} \}_{j=1\,\text{to}\,4} \gets \textsc{Spectrum}(\phi_\mathrm{diag})$
				\Comment $\lambda_j\geqslant \lambda_{j+1}$\;
				\State $\phi_\mathrm{extra} \gets (\lambda_1 - \lambda_2) \ketbra{\phi_2}{\phi_2}$\;
				\State $\sigma' \gets \sigma_\mathrm{diag}+ \phi_\mathrm{extra}$
				\Comment Add extra noise
				\State $\{\theta^*, i^* \}\gets \textsc{DeCrossCorrelate}(\sigma')$
				\Comment Find~$\theta^*\in[0, \pi]$ and mode~$i^*\in\{1, 2\}$\;
				\State $\sigma_\mathrm{rot} \gets {K_\mathrm{rot}^{i^*}}(\theta^*) \, \sigma' \, {K_\mathrm{rot}^{i^*}}^\top(\theta^*)$
				\Comment De-cross-correlate by rotating mode~$i^*$ by~$\theta^*$
				\State $\sigma_3\gets\textsc{AddMode3}(\sigma_\mathrm{rot})$
				\Comment Add third ancillary mode\;
				\State $\{\tau^*, j^* \} \gets \max_{\tau, j} \textsc{EOF}\left[K_\mathrm{bs}^{3,j} (\tau) \sigma_3 \left ({K_\mathrm{bs}^{3,j}}(\tau)\right)^\top\right]$
				\Comment Find $\tau^* \in [0, 1]$ to mix mode~$j^* \in \{1, 2\}$ with mode~3\;
				\State$ \sigma_\mathrm{out}\gets \textsc{RemoveMode3}\left( K_\mathrm{bs}^{3, j^*}(\tau^*) \sigma_3 \left ({K_\mathrm{bs}^{3,j^*}}(\tau^*)\right)^\top \right)$ \label{step:13}
				\Comment Mix modes and discard ancillary third mode
				\EndProcedure
			\end{algorithmic}
		\end{algorithm}

		We note that for states~$\sigma_\mathrm{in}$
		with both modes squeezed,
		steps~\ref{step:5} through to~\ref{step:13} may be skipped in
		Alg.~\ref{algorithm:MaximizeEOF},
		and instead a final balanced beam splitter~$K_\mathrm{bs}$ suffices to bring~$\sigma_\mathrm{in}$ into~$G$.
		More precisely,
		\begin{equation}
			\sigma_\mathrm{out} = K_\mathrm{bs} K_\mathrm{BM} \sigma_\mathrm{in} K_\mathrm{BM}^\top K_\mathrm{bs} ^\top \in G.
		\end{equation}
		Thus~$K_\mathrm{bs} K_\mathrm{BM}$ is the passive transformation that maximizes the EOF of~$\sigma_\mathrm{in}$,
		or, alternatively, transforms~$\sigma_\mathrm{in}$ into the set~$G$.

		\subsection{Numerical Simulations}
		\label{sec:results}
		%\subsection{Analytical Results}

		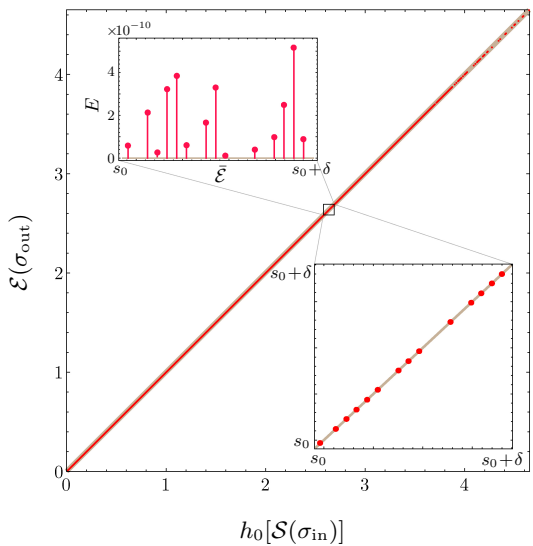
\begin{figure}[btp]
			\centering
			\begin{tikzpicture}
				\node [anchor=south west,inner sep=0] (image) at (0,0) {\includegraphics[width=180pt]{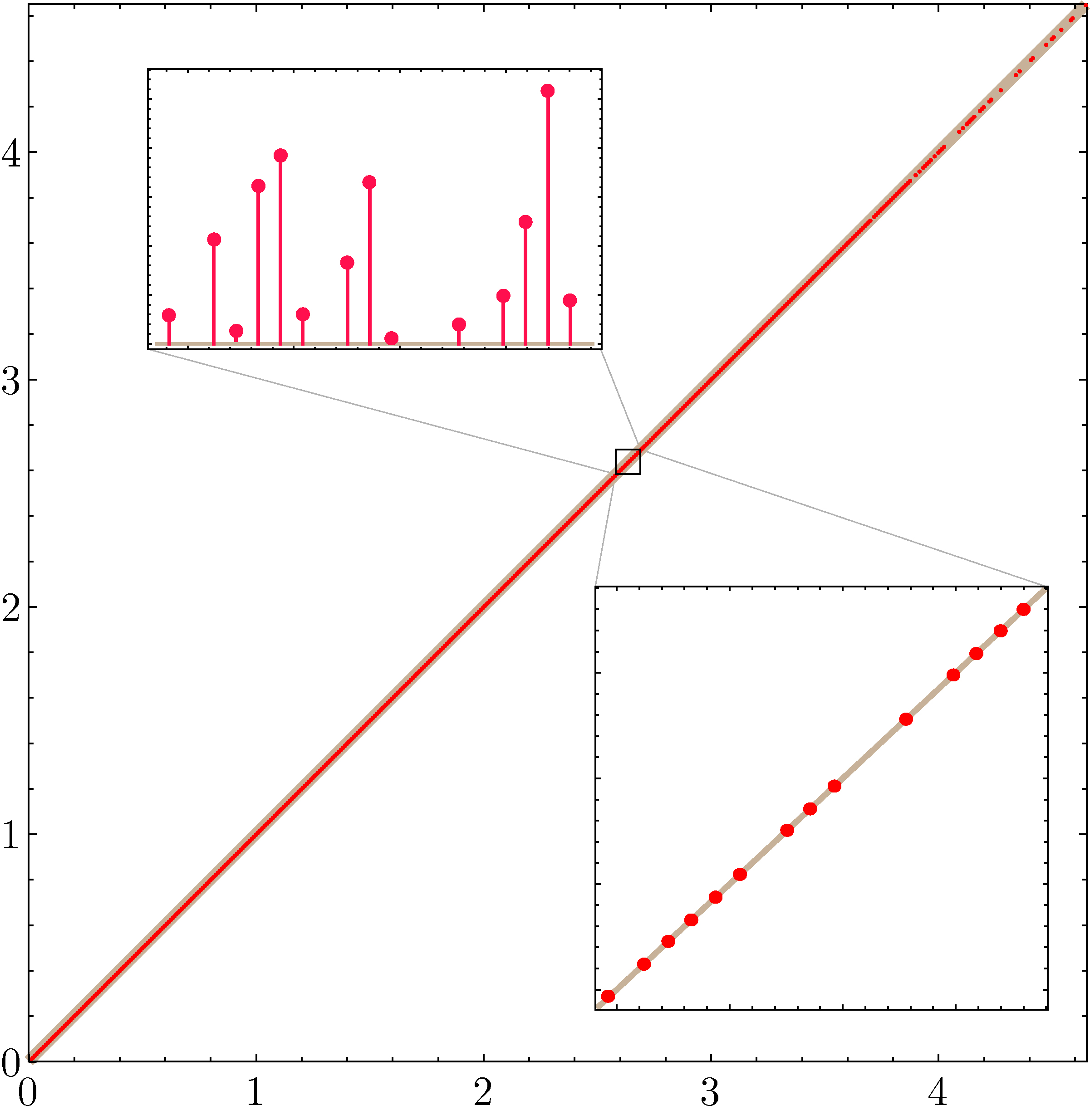}};
				\begin{scope}[
					x={(image.south east)},
					y={(image.north west)}
					]
					\node [black] at (0.5, -0.07) {$h_0[\mathcal{S}(\sigma_\mathrm{in})]$};
					\node [black,rotate=90] at (-0.07, 0.5) {$\mathcal{E}(\sigma_\mathrm{out})$};
					\node [black,scale=0.8] at (0.35, 0.66) {$\bar{\mathcal{E}}$};
					\node [black,rotate=90,scale=0.8] at (0.08, 0.8) {$ E$};
					\node [black,scale=0.6] at (0.12, 0.692) {$0$};
					\node [black,scale=0.6] at (0.12, 0.78) {$2$};
					\node [black,scale=0.6] at (0.12, 0.87) {$4$};
					\node [black,scale=0.6] at (0.16, 0.956) {$\!\times\!10^{-10}$};
					%\node [black,scale=0.8] at (0.618, 0.7) {$\!\times\!10^{-6}$};
					\node [black,scale=0.7] at (0.14, 0.668) {$s_0$};
					\node [black,scale=0.7] at (0.54, 0.67) {$s_0\!+\!\delta$};
					\node [black,scale=0.75] at (0.52, 0.11) {$s_0$};
					\node [black,scale=0.75] at (0.555, 0.076) {$s_0$};
					\node [black,scale=0.75] at (0.497, 0.458) {$s_0\!+\!\delta$};
					\node [black,scale=0.75] at (0.94, 0.076) {$s_0\!+\!\delta$};
				\end{scope}
			\end{tikzpicture}
			\caption{Numerical results from running Alg.~\ref{algorithm:MaximizeEOF} on a million random two-mode Gaussian states.
				The output state's~$\mathcal{E}$ and the input state's~$h_0(\mathcal{S})$ values coincide (red dots) and, thus,
				lie on the~$Y=X$ line (thick, gray) to within numerical tolerance. The bottom inset magnifies the section
				$[s_0, s_0+\delta]$ (where~$s_0=2.6430777$ and~$\delta=4.1\!\times\!10^{-6}$) of the main plot.
				The top inset rotates this same section, by plotting the error~$E=\mathcal{E}-h_0[\mathcal{S}]$ against~$\bar{\mathcal{E}}=(\mathcal{E}+h_0[\mathcal{S}])/2$. Over a million runs, the average absolute error~$\vert E \vert_{\mathrm{avg}}$ is~$1.93\!\times\!10^{-9}$.%with standard deviation~$8.03\!\times\!10^{-9}$
			}
			\label{fig:3NumericalVerification}
		\end{figure}

		%\subsection{Numerical Investigation}
		In order to support our conjecture,
		we numerically apply Alg.~\ref{algorithm:MaximizeEOF}
		to $10^6$ randomly generated two-mode Gaussian states.
		This random generation leverages Williamson's decomposition
		(see Appendix~\ref{sec:appB})
		%to produce arbitrary Gaussian states
		by applying random active and passive operations on
		randomly generated two-mode thermal states.
		For each randomly generated instance,
		its SOF and the corresponding optimum pure state is computed numerically, based on an algorithm provided in Ref.~\cite{idel2016operational} with a numerical accuracy of $10^{-8}$.
		Then, this state is transformed according to Alg.~\ref{algorithm:MaximizeEOF},
		and the EOF of the output state is calculated.
		For arbitrary two-mode Gaussian states,
		there are several equivalent approaches (but no simple analytical expression)
		to calculate the Gaussian EOF~\cite{WGKW+04,MM08,TOR19,AI05,IS08}.
		We used the approach from Ref.~\cite{AI05} to compute Gaussian EOFs in this work.

		By testing on~$10^6$ such randomly generated two-mode Gaussian states,
		we see that the difference between the
		EOF~$\mathcal{E}(\sigma_\mathrm{out})$
		and the upper
		bound~$h_0\left[\mathcal{S}\left(\sigma_\mathrm{in}\right)\right]$ is
		always lower than numerical tolerance.  The average absolute
		error~$\left\vert \mathcal{E}-h_0[\mathcal{S}] \right\vert$ over a million runs
		is~$1.93\!\times\!10^{-9}$.

		We also explicitly verify that Alg.~\ref{algorithm:MaximizeEOF}
		does not change the SOF of the input state,
		i.e.,
		$\mathcal{S}(\sigma_\mathrm{in}) = \mathcal{S}(\sigma_\mathrm{out})$.
		The results from this test are shown in Fig.~\ref{fig:3NumericalVerification},
		where the straight line plot between~$\mathcal{E}$ and~$h_0[\mathcal{S}]$
		provides strong evidence supporting our conjecture.

		Based on our proposition,
		and the numerical results supporting our conjecture shown in Fig.~\ref{fig:3NumericalVerification},
		%of Alg.~\ref{algorithm:MaximizeEOF},
		it follows that the EOF potential of all
		two-mode Gaussian states is a monotonic function of the state's SOF.
		Qualitatively, this means the maximum EOF, when restricted to linear
		passive optics, is completely determined by the
		minimum amount of local squeezing required for state preparation.
		Conversely, to increase EOF beyond this value, further squeezing
		operations are necessarily required.

		\section{Discussion}\label{conclusion}
		In this work, we have studied the relation between the squeezing of formation and the maximum entanglement of formation under passive operations for two-mode Gaussian states.
		We have characterised a special six-parameter family of two-mode states,
		which are potential-saturating and also saturate the SOF-EOF bound.
		Moreover, we have conjectured that any arbitrary two-mode Gaussian state
		can be passively transformed into the aforementioned family.
		In support of our conjecture, we have proposed an algorithm
		to passively transform arbitrary two-mode Gaussian states into this special class.
		Finally, we report numerical results from
		simulating this algorithm on a million random instances,
		which supports our conjecture.

		In conclusion, we claim that the entanglement potential
		for all two-mode Gaussian states
		is completely determined by the
		minimum amount of squeezing
		required to construct the state.
		By connecting an operational measure for squeezing
		to one for entanglement,
		our work establishes a satisfying
		link between the resource theories of squeezing and entanglement.
		Furthermore, being restricted solely to passive linear optics,
		the steps in our proposed algorithm are practically feasible in experimental setups.
		As an example application, our results could be used to
		quantify and compare the entangling capabilities of different experimental setups.

		Our work draws a natural conclusion to
		the line of research investigating the relationship
		between entanglement potential and squeezing
		for two-mode Gaussian states.
		As both these quantities can be
		extended to multi-mode states,
		the validity of the SOF-EOF bound
		and its saturability remain open problems in the greater-than-two-mode case.
		Notably, in this case, the Gaussian EOF
		and the EOF do not coincide so
		the entanglement potential must be redefined carefully~\cite{akbari2015entanglement,WGKW+04}.

		\section*{DATA AVAILABILITY}
		The data that support the findings of this study are
		available from the corresponding author upon reasonable
		request.

		\section*{CODE AVAILABILITY}
		The codes that support the findings of this study are
		available from the corresponding author upon reasonable
		request.

		\section*{Acknowledgements}
		The authors are grateful to Prof.\ Tim Ralph for helpful discussions.
		This research is supported by the Australian Research
		Council (ARC) under the Centre of Excellence for Quantum
		Computation and Communication Technology~CE170100012.
		Authors S.T.\ and P.N.\ acknowledge support for this work from the National Science Foundation under grant number NSF CNS 2106887 on ``U.S.-Ireland R\&D Partnership: Collaborative Research: CNS Core: Medium: A unified framework for the emulation of classical and quantum physical layer networks'' and the NSF QuIC-TAQS program ``QuIC-TAQS: Deterministically Placed Nuclear Spin Quantum Memories for Entanglement Distribution'' under grant number NSF OMA 2137828.

		\section*{Author contributions statement}
		S.T.\ and S.M.A.\ conceived the project.
		B.L., S.T., A.D.\ and S.M.A.\ developed the theory
		and performed the numerical simulation.
		A.D.\ wrote the manuscript.
		All authors contributed to discussions regarding the results of the paper.
		P.K.L.\ supervised the project.

		\section*{Additional information}
		The Authors declare no Competing Financial or Non-Financial Interests.

		\appendix

		\section{Gaussian transformations and their matrix representations}
		\label{sec:appA}
		In this section, we briefly review some common Gaussian transformations.
		Gaussian transformations are represented by symplectic matrices~$S$,
		which satisfy~$S \Omega S^\top = \Omega$,
		where
		\begin{equation}
			\Omega =  \bigoplus_{j=1}^N \begin{pmatrix} 0  & 1   \\ -1 & 0\end{pmatrix}
		\end{equation}
		is the symplectic form.
		Two-mode beam splitters are represented by
		\begin{equation}
			K_\mathrm{bs}(\tau) = \begin{pmatrix} \sqrt\tau \; \mathds{1}_2  & \sqrt{1-\tau}  \; \mathds{1}_2 \\ -  \sqrt{1-\tau}  \; \mathds{1}_2 &\sqrt\tau \;  \mathds{1}_2\end{pmatrix},
		\end{equation}
		where~$\mathds{1}_N$ represents the identity operator on~$N$ dimensions.
		For balanced beam splitters, where~$\tau=1/2$, we drop the explicit dependence on~$\tau$ and simply write~$K_\mathrm{bs}$.
		Single-mode phase rotations are represented via the rotation matrices,
		\begin{equation}
			K_\mathrm{rot}(\theta) = \begin{pmatrix} \cos\theta  & \sin \theta \\ - \sin \theta & \cos\theta \end{pmatrix}.
		\end{equation}
		Beam splitters and phase rotations constitute the fundamental passive linear transformations,
		as can be seen, e.g., from the rectangular decomposition for passive Gaussian operations (see Sec.~\ref{sec:appB}).

		A single-mode squeezer is represented by
		\begin{equation}
			S(r) = \begin{pmatrix}  e^{-r}  & 0 \\  0 & e^{r} \end{pmatrix},
		\end{equation}
		whereas a two-mode squeezer can either be local, as in~$S_1(r_1, r_2)\coloneqq S(r_1) \oplus S(-r_2)$,
		wherein two single-mode squeezers act independently on two modes indexed by~$j\in\{1,2\}$,
		or non-local, as in
		\begin{equation}
			\begin{aligned}
			S_2(r) & \coloneqq  K_\mathrm{bs} \;  S_1(-r, r) \;  K_\mathrm{bs}^\top  \\
			&= K_\mathrm{bs} \begin{pmatrix}  e^{r}  & 0 &0  & 0 \\  0 & e^{-r}  & 0 & 0  \\ 0 &0  &e^{-r} &0 \\  0 &0 & 0 & e^{r}   \end{pmatrix} K_\mathrm{bs}^\top \,.
			\end{aligned}
		\end{equation}
		The locally-squeezed two-mode diagonal Gaussian state~$\pi_\mathrm{d}$ is
		\begin{equation}
			\begin{aligned}
			\pi_\mathrm{d}(r_1, r_2) &= S_1(r_1, r_2) \mathds{1}_4 S_1(r_1, r_2)^\top\\
			&=  \begin{pmatrix}  e^{-2 r_1}  & 0 &0  & 0 \\  0 & e^{2 r_1}  & 0 & 0  \\ 0 &0  &e^{2 r_2} &0 \\  0 &0 & 0 & e^{-2 r_2}   \end{pmatrix}.
			\end{aligned}
		\end{equation}
		For convenience, we sometimes drop the~$r_1, r_2$ dependence and write simply~$\pi_\mathrm{d}$.
		When restricted to two-mode Gaussian states,
		the Bloch-Messiah decomposition~\cite{dutta1995real,braunstein2005squeezing}
		states that
		phase rotations, beam splitters and single- and two-mode squeezers
		are sufficient to implement arbitrary Gaussian transformations.

		Finally, for a special subset of two-mode Gaussian states called
		de-cross-correlated states~\cite{AGLL20,tserkis2020maximum} that
		satisfy~$\langle \hat{x}_i \hat{p}_j + \hat{p}_i
		\hat{x}_j \rangle =0$ for~$i, j \in \{1, 2\}$, the covariance
		matrix~$\sigma$ takes the form
		\begin{equation}
			\begin{aligned}
			\sigma &= \begin{pmatrix}\sigma_{11} &0 &\sigma_{13} &0 \\
				0 &\sigma_{22} &0 &\sigma_{24} \\
				\sigma_{13} &0 &\sigma_{33} &0 \\
				0 &\sigma_{24} &0 &\sigma_{44} \end{pmatrix}  \\
			&= \begin{pmatrix} \sigma_{11} & \sigma_{13} \\
				\sigma_{13}&\sigma_{33} \end{pmatrix} \oplus \begin{pmatrix} \sigma_{22} & \sigma_{24} \\
				\sigma_{24}&\sigma_{44} \end{pmatrix} \eqqcolon C_q \oplus C_p,
			\end{aligned}
			\label{eq:decrosscorrelatedstandardform}
		\end{equation}
		where in the second equality we have re-indexed the
		usual~$[\hat{x}_1, \hat{p}_1, \hat{x}_2, \hat{p}_2]$ operators
		as~$[\hat{x}_1, \hat{x}_2, \hat{p}_1, \hat{p}_2]$.  Note that the form
		in Eq.~\eqref{eq:decrosscorrelatedstandardform} is also referred to as
		a standard form~\cite{DGCZ00} and a large class of two-mode Gaussian
		states can be passively transformed into this form~\cite{AGLL20}.
		Moreover, for pure states in this standard form, $C_p = C_q^{-1}$ (the
		superscript~$^{-1}$ denoting matrix inverse), so
		that~$\pi=C_q \oplus C_q^{-1}$ and consequently,~$\mathcal{E} (\pi)$
		becomes a monotonically increasing function of the sub-correlation
		matrix~$C_q$~\cite{AI05,WGKW+04}.

		\section{Standard decompositions in Gaussian optics}
		\label{sec:appB}
		In this section, we briefly overview some mathematical decompositions of Gaussian states and operations.
		The first, known as Williamson's decomposition,
		provides a way to write arbitrary mixed state~$\sigma$ as a Gaussian (symplectic) transformation~$S$
		performed on a thermal state represented by diagonal matrix~$D$, i.e.,
		\begin{equation}
			\sigma = S D S^\top,
		\end{equation}
		where
		\begin{equation}
			%	\begin{split}
				D \coloneqq \mathrm{Diag} \left (\nu_1, \nu_1,
				\dots, \nu_N, \nu_N\right ),
			\end{equation}
			and
			\begin{equation}
				\{\nu_j \} \coloneqq  \left\vert \mathrm{Eigs}\left ( i \Omega \sigma \right ) \right\vert
				%	\end{split}
		\end{equation}
		are the symplectic eigenvalues.
		Here,~$\mathrm{Diag}( x )$ denotes a diagonal matrix with diagonal
		entries given by~$x$ and~$\mathrm{Eigs}(A)$ denotes the eigenvalues of
		a matrix $A$.

		The second decomposition, known as the Bloch-Messiah decomposition
		or Euler decomposition,
		presents a way to decompose arbitrary Gaussian (symplectic) transformation~$S$
		into a sequence of a passive (orthogonal) transformation~$K_1$,
		followed by single-mode squeezing operations on each mode
		(represented by diagonal matrix~$Z=\oplus_{j=1}^N S(r_j)$),
		and a second passive transformation~$K_2$,
		i.e.,
		\begin{equation}
			S = K_1 Z K_2  \, .
		\end{equation}

		One useful application of the Bloch-Messiah decomposition
		is in passively diagonalising a pure Gaussian~$\pi$.
		Being a pure state,~$\pi$ can be written as~$\pi = S S^\top$
		for~$S$ some symplectic matrix.
		The Bloch-Messiah decomposition then yields~$S=K_1 Z K_2$,
		which results in
		\begin{equation}
			\pi= K_1 Z Z^\top K_1^\top = K_1 Z^2 K_1^\top.
		\end{equation}
		Evidently, ~$K_1^\top \pi K_1=Z^2$, which is a diagonal matrix, so
		the passive operation that diagonalises~$\pi$ is~$K_1^\top$.

		Third,
		the rectangular decomposition allows us to
		decompose arbitrary passive (i.e., real, orthogonal, and symplectic) transformations
		into as a sequence of~$N(N-1)/2$ beam splitters and single-mode phase rotations.
		Mathematically,
		this implies that beam-splitters and single-mode phase rotations generate the set of passive Gaussian transformations.

		Lastly,
		the polar decomposition equates any symplectic matrix~$S$
		with the product of an orthogonal, symplectic matrix~$K$
		(representing a passive operation)
		and a symmetric, positive semidefinite, symplectic matrix~$P$,
		i.e.,
		\begin{equation}
			S = K P.
		\end{equation}
		This decomposition is unique and plays a crucial role in Gaussian quantum optics~\cite{idel2016operational}.
		\section{De-cross-correlated pure states saturating the EOF Potential}
		\label{sec:appC}
		In this section,
		we establish necessary and sufficient conditions
		for a de-cross-correlated pure state~$\pi_\mathrm{dcc}$
		to be potential-saturating,
		i.e.,~$\mathcal{E}(\pi_\mathrm{dcc})=\mathcal{P}(\pi_\mathrm{dcc})$.
		In the process,
		we also compute the SOF of this state
		and find that the potential-saturating state also saturates
		the upper bound
		\begin{equation}
			\mathcal{P}(\sigma) \leq h_0\left [\mathcal{S}(\sigma) \right ],
			\label{eq:defBOUNDSupp}
		\end{equation}
		meaning that
		\begin{equation}
			\mathcal{E}(\pi_\mathrm{dcc})=\mathcal{P}(\pi_\mathrm{dcc})= h_0\left [\mathcal{S}(\pi_\mathrm{dcc}) \right ]
		\end{equation}
		when~$\pi_\mathrm{dcc}$ satisfies the necessary and sufficient conditions.

		As stated previously, the covariance matrix
		of a de-cross-correlated pure state~$\pi_\mathrm{dcc}$
		takes the form
		\begin{equation}
			\pi_\mathrm{dcc}  = C_q\oplus {C_q}^{-1} = \begin{pmatrix}q_1 & q_3\\q_3 & q_2\end{pmatrix}\oplus\begin{pmatrix}q_1 & q_3\\q_3 & q_2\end{pmatrix}^{-1},
			\label{eq:dccpurestate}
		\end{equation}
		and the uncertainty relation for~$\pi_\mathrm{dcc}$ simplifies to~$q_1q_2-q_3^2\geq0$ and~$q_1, q_2>0$.
		As shown in the discussion on the Bloch-Messiah decomposition
		in Sec.~\ref{sec:appB},
		any pure state can be passively transformed into a diagonal state,
		and hence a de-cross-correlated state.
		%Thus, there always exists some~$\pi_\mathrm{dcc}$ with the same EOF and SOF properties as the EOF-optimal and the SOF-optimal pure states, respectively, for any given mixed state~$\sigma$.

		For the de-cross-correlated pure state in Eq.~\eqref{eq:dccpurestate},
		the two largest eigenvalues~$\lambda_3^{\uparrow}$ and~$\lambda_4^{\uparrow}$
		can be computed as
		\begin{equation}
			\begin{aligned}
			\lambda_3^{\uparrow} &= \frac{q_1+q_2+\sqrt{(q_1-q_2)^2+4q_3^2}}{2},  \\
			\text{and} \qquad \lambda_4^{\uparrow} &= \frac{q_1+q_2+\sqrt{(q_1-q_2)^2+4q_3^2}}{2(q_1q_2-q_3^2)},
			\end{aligned}
		\end{equation}
		so that the SOF of the state is
		\begin{equation}
			\begin{aligned}
			\mathcal{S}(\pi_\mathrm{dcc})  &= \frac{1}{2} \ln (\lambda_3^{\uparrow}\lambda_4^{\uparrow})
			= \ln \left ( \frac{q_+ +\sqrt{{q_+}^2-4\Delta}}{2\sqrt{\Delta}} \right )\\
			&=\ln\left (\sqrt{\frac{{q_+}^2}{4q_1q_2}\frac{q_1q_2}{\Delta}}+\sqrt{\frac{{q_+}^2}{4q_1q_2}\frac{q_1q_2}{\Delta}-1}\right ) .
			\end{aligned}
			\label{eq:important}
		\end{equation}
		Here we have defined~$q_+ \coloneqq q_1+q_2$
		and~$\Delta \coloneqq \mathrm{det} (C_q)$ in the second equality
		and rearranged terms in the third.
		Note that~$\mathcal{S}(\pi_\mathrm{dcc})$ is a monotonically increasing function of~$q_+$.
		Using the AM-GM inequality $q_+ \geq 2\sqrt{q_1 q_2}$ and
		the monotonicity of logarithms,
		we get the inequality
		\begin{equation}
			\mathcal{S}(\pi_\mathrm{dcc}) \geq \text{ln} \left (\sqrt{\frac{q_1q_2}{\Delta}}+\sqrt{\frac{q_1q_2}{\Delta}-1} \,\right ) ,
			\label{eq:Sbound}
		\end{equation}
		where the equality can hold if and only if~$q_1=q_2$.
		Next we define~$m(C_q) \coloneqq q_1 q_2/ \Delta$,
		so that Eq.~\eqref{eq:Sbound} becomes
		%defining the right hand side of Eq.~\eqref{eq:Sbound} as
		\begin{equation}
			\mathcal{S}(\pi_\mathrm{dcc}) \geq \ln \left ( \sqrt{m(C_q)}+\sqrt{m(C_q)-1} \right ) \eqqcolon \mathcal{S}_0(\pi_\mathrm{dcc}) .
			\label{eq:Spuredccgreater}
		\end{equation}
		It is now evident that~$\mathcal{S}_0(\pi_\mathrm{dcc})$ is a monotonically increasing function of $m(C_q)$
		and lower-bounds~$\mathcal{S}(\pi_\mathrm{dcc})$, being equal if and only if~$q_1 = q_2$.

		Consider, on the other hand,
		that the EOF~$\mathcal{E}(\pi_\mathrm{dcc})$
		of a pure de-cross-correlated state
		is also a monotonically increasing function of~$m(C_q)$~\cite{AI05,WGKW+04}, given by
		\begin{equation}
			\mathcal{E}{(\pi_\mathrm{dcc})}
			= h\left [\sqrt{m(C_q)}-\sqrt{m(C_q)-1} \right ],
			\label{eq:eofpdccpurestate}
		\end{equation}
		where
		\begin{equation}
			%	\begin{split}
				h[x]{\coloneqq}\frac{(1+x)^2}{4x} \ln \left ( \frac{(1+x)^2}{4x}\right ){-}\frac{(1-x)^2}{4x} \ln \left ( \frac{(1-x)^2}{4x} \right ) .
				%	\end{split}
		\end{equation}
		and we also define~$h_0[x]\coloneqq h[e^{-x}]$.
		It then follows from Eqs.~\eqref{eq:Spuredccgreater} and~\eqref{eq:eofpdccpurestate} that
		\begin{equation}
			\begin{aligned}
			\mathcal{E}(\pi_\mathrm{dcc}) &=h \left [ e^{\ln \left (\sqrt{m(C_q)}-\sqrt{m(C_q)-1} \right )} \right ]\\
			&= h \left [ e^{-\ln \left (\sqrt{m(C_q)}+\sqrt{m(C_q)-1} \right )} \right ]\\
			&= h_0 \left [\mathcal{S}_0(\pi_\mathrm{dcc}) \right ] \leq h_0 \left [ \mathcal{S}(\pi_\mathrm{dcc})\right ]
			\end{aligned}
		\end{equation}
		for all pure de-cross-correlated states,
		with equality holding if and only if~$q_1=q_2$.
		The $q_1=q_2$ condition physically corresponds to
		the state from a balanced beam splitter mixing two non-correlated modes,
		which results in a state with two identical modes.
		Combining $\mathcal{E}(\pi_\mathrm{dcc})= h_0 \left [ \mathcal{S}(\pi_\mathrm{dcc})\right ] \iff q_1 = q_2$ with~$\mathcal{E}(\pi_\mathrm{dcc})\leq \mathcal{P}(\pi_\mathrm{dcc}) \leq h_0 \left [ \mathcal{S} (\pi_\mathrm{dcc}) \right ]$,
		we conclude that a pure de-cross-correlated state
		can satisfy~$\mathcal{E}(\pi_\mathrm{dcc}) =  \mathcal{P}(\pi_\mathrm{dcc})$
		if and only if
		$q_1=q_2$.
		We also conclude that this special class of de-cross-correlated states
		satisfy~$\mathcal{P}(\pi_\mathrm{dcc}) = h_0 \left [ \mathcal{S} (\pi_\mathrm{dcc}) \right ]$ and thus saturate the upper bound in Eq.~\eqref{eq:defBOUNDSupp}.

		The above discussion suggests a way to maximise a given state's EOF via passive operations.
		%In other words, the EOF potential $\mathcal{P}(\pi)$ and the auxiliary function $h_0\left [\mathcal{S}_0(\pi) \right ]$ are shown to be equivalent and directly related to $m(C_q)$.
		For mixed states, both the EOF and the SOF are defined via a convex optimisation
		over all possible pure state decompositions, so a mixed state $\sigma$ could saturate the upper bound in Eq.~\eqref{eq:defBOUNDSupp} if its potential-saturating pure state $\pi_\mathrm{opt}$ is de-cross-correlated (i.e., of the form in Eq.~\eqref{eq:dccpurestate}) and has~$q_1=q_2$.
		This is the fundamental idea behind our algorithm,
		which first performs passive operations
		required to decouple the squeezing between two modes of a given input state,
		then de-cross-correlates this state passively,
		and finally, by appropriately mixing with an ancillary vacuum,
		removes any excess noise in order to maximise the EOF.

				% Bibliography
		%\bibliographystyle{unsrt}
		\bibliography{mybibtex2} %bibliography

	\end{document}